\documentclass[onefignum,onetabnum]{siamart171218}
\usepackage{bm}

\usepackage{lipsum}
\usepackage{amsfonts}
\usepackage{graphicx}
\usepackage{epstopdf}
\usepackage{algorithmic}
\ifpdf
  \DeclareGraphicsExtensions{.eps,.pdf,.png,.jpg}
\else
  \DeclareGraphicsExtensions{.eps}
\fi

\newsiamremark{remark}{Remark}
\newsiamremark{hypothesis}{Hypothesis}
\crefname{hypothesis}{Hypothesis}{Hypotheses}
\newsiamthm{claim}{Claim}

\headers{Graph distance from length spectrum}{L. Torres, P. Su\'arez-Serrato, T. Eliassi-Rad }

\title{Graph Distance from the Topological View of Non-Backtracking Cycles}

\author{Leo Torres\thanks{Network Science Institute, Northeastern University, Boston, MA 
  (\email{leo@leotrs.com}).}
\and Pablo Su\'arez-Serrato\thanks{Department of Mathematics, UC Santa Barbara USA, and Instituto de Matem\'aticas, Universidad Nacional Aut\'onoma de M\'exico, Ciudad de M\'exico, CDMX, (\email{pablo@im.unam.mx}).}
\and Tina Eliassi-Rad\thanks{Network Science Institute \& College of Computer and Information Science, Northeastern University, Boston, MA, (\email{tina@eliassi.org}).}}

\usepackage{amsopn}

\ifpdf
\hypersetup{
  pdftitle={Graph Distance from the Topological View of Non-Backtracking Cycles},
  pdfauthor={L. Torres, P. Su\'arez Serrato, and T. Eliassi-Rad}
}
\fi

\begin{document}
\maketitle

% REQUIRED
\begin{abstract}
Whether comparing networks to each other or to random expectation, measuring dissimilarity is essential to understanding the complex phenomena under study. However, determining the structural dissimilarity between networks is an ill-defined problem, as there is no canonical way to compare two networks. Indeed, many of the existing approaches for network comparison differ in their heuristics, efficiency, interpretability, and theoretical soundness. Thus, having a notion of distance that is built on theoretically robust first principles and that is interpretable with respect to features ubiquitous in complex networks would allow for a meaningful comparison between different networks. Here we introduce a theoretically sound and efficient new measure of graph distance, based on the ``length spectrum" function from algebraic topology, which compares the structure of two undirected, unweighted graphs by considering their non-backtracking cycles. We show how this distance relates to structural features such as presence of hubs and triangles through the behavior of the eigenvalues of the so-called non-backtracking matrix, and we showcase its ability to discriminate between networks in both real and synthetic data sets. By taking a topological interpretation of non-backtracking cycles, this work presents a novel application of Topological Data Analysis to the study of complex networks.
\end{abstract}

% REQUIRED
\begin{keywords}
  graph distance, algebraic topology, length spectrum
\end{keywords}

% REQUIRED
\begin{AMS}
  Spectral Graph Theory, Length Spectrum, Random Graphs, Metric Spaces, Topological Data Analysis, Geometric Data Analysis 
\end{AMS}

\section{Introduction}\label{sec:intro}
As the Network Science literature continues to expand and scientists compile more and more examples of real life networked data sets \cite{icon,kunegis2013konect} coming from an ever growing range of domains, there is a need to develop methods to compare complex networks, both within and across domains. Many such graph distance measures have been proposed \cite{soundarajan2014sdm,koutra2013deltacon,bagrow2018information,bento2018family,PhysRevE.86.036104,schieber2017quantification,chowdhury2017distances,chowdhury2018metric,berlingerio2013asonam}, though they vary in the features they use for comparison, their interpretability in terms of structural features of complex networks, and their computational costs, as well as in the discriminatory power of the resulting distance measure. This reflects the fact that complex networks represent a wide variety of systems whose structure and dynamics are difficult to encapsulate in a single distance score. For the purpose of providing a principled, interpretable, efficient and effective notion of distance, we turn to the \emph{length spectrum} function, which can be defined on a broad class of metric spaces that includes Riemannian manifolds as well as graphs. The discriminatory power of the Length spectrum is well known in other contexts: it can distinguish certain one-dimensional metric spaces up to isometry \cite{ConstantineLafont}, and it determines the Laplacian spectrum in the case of closed hyperbolic surfaces \cite{leininger2007length}. However, it is not clear if this discriminatory power is also present in the case of complex networks. Accordingly, we present a study on the following question: \textbf{is the length spectrum function useful for the comparison of complex networks?}

We answer the above question in the positive by introducing the \textbf{Truncated Non-Backtracking Spectral Distance (TNBSD)}:  a principled, interpretable, efficient, and effective method that quantifies the distance between two undirected, unweighted networks. \textbf{TNBSD} has several desirable properties.  First, \textbf{TNBSD} is based on the theory of the length spectrum and the set of non-backtracking cycles of a graph (a non-backtracking cycle is a closed walk that does not retrace any edges immediately after traversing them); these provide the theoretical background of our method. Second, \textbf{TNBSD} is interpretable in terms of features of complex networks  such as existence of hubs and triangles.  This helps in the interpretation and visualization of distance scores yielded by \textbf{TNBSD}. Third, \textbf{TNBSD} is a computationally efficient method, needing no more than the computation of a few largest eigenvalues of the so-called non-backtracking matrix of a graph, Fourth, \textbf{TNBSD} is effective at distinguishing real and synthetic networks as shown by our extensive experiments in Section~\ref{sec:experiments}. In studying \textbf{TNBSD}, we highlight the topological interpretation of the non-backtracking cycles of a graph, present an efficient algorithm to compute the non-backtracking matrix, and discuss the data visualization capabilities of its complex eigenvalues (see  Fig.~\ref{fig:random_eigenvalues}).

\begin{figure}[t]
\centering
\includegraphics[width=0.6\textwidth]{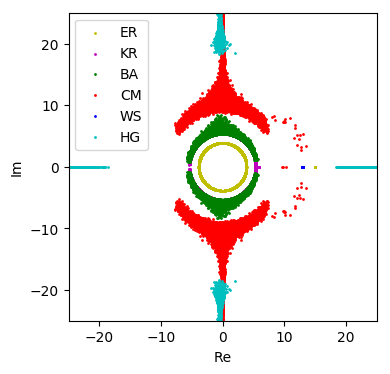}
\caption{\label{fig:random_eigenvalues}
(Best viewed in color.) Complex eigenvalues of the non-backtracking matrix of random graphs (see Sec.~\ref{sub:nbc} for definition). For different random graph models --Erd\"os-R\'enyi (ER) \cite{erdos1960p,bollobas2001random}, Barab\'asi-Albert (BA) \cite{barabasi1999emergence}, Stochastic Kronecker Graphs (KR) \cite{leskovec2010kronecker,seshadhri2013depth}, Configuration Model with power law degree distribution (CM; $\gamma=2.3$) \cite{newman2003structure}, Watts-Strogatz (WS) \cite{watts1998collective}, Hyperbolic Graphs (HG; $\gamma=2.3$) \cite{krioukov2010hyperbolic,aldecoa2015hyperbolic}-- we plot the largest $r=200$ eigenvalues of each of 50 random graphs of each model on the complex plane. Observe that each model generates eigenvalue distributions presenting different geometric patterns. We analyze and exploit these patterns in order to fine-tune the Truncated Non-backtracking Spectral Distance, TNBSD (See Sec.~\ref{sec:nbm}). To make the plot more readable,  we do not show all of the eigenvalues. All graphs have $n=5\times10^4$ nodes and average degree approximately $\langle k \rangle = 15$.}
\end{figure}

\paragraph{Perspective} Hashimoto \cite{hashimoto1989zeta} discussed the non-backtracking cycles of a graph (and the associated non-backtracking matrix) in relation to the theory of Zeta functions in graphs. Terras \cite{terras2010zeta} explained the relationship between them and the free homotopy classes of a graph (see Sec.~\ref{sec:background}). More recently, the non-backtracking matrix has been used in the Network Science literature for diverse applications such as node centrality \cite{martin2014localization} and community detection \cite{krzakala2013spectral}, and the data mining tasks of clustering \cite{ren2011graph} and embedding \cite{jiang2018spectral}. In particular, the application to community detection is of special interest since it was proven that the non-backtracking matrix performs better at spectral clustering than the Laplacian matrix in some cases \cite{krzakala2013spectral}. Hence, there is recent interest in describing the eigenvalue distribution of the non-backtracking matrix in models such as the Erd\"os-R\'enyi random graph and the stochastic block model \cite{bordenave2015non,gulikers2016non,wood2017limiting,saade2014spectral}. %Doing this with more precision, and in the general case, remains an open problem. 
 Our work differs from other applied treatments of the non-backtracking matrix in that we arrive at its eigenvalues from first principles, as a relaxed version of the length spectrum. Concretely, we use the eigenvalues to compare graphs because the spectral moments of the non-backtracking matrix describe certain aspects of the length spectrum (see Sec.~\ref{sec:relaxation}). The spectral moments of the adjacency and Laplacian matrices are also known to describe certain structural features of networks \cite{estrada1996spectral,preciado2013structural}. 

The rest of this paper is structured as follows. Section~\ref{sec:background} provides necessary background information on the length spectrum, non-backtracking cycles, and the non-backtracking matrix. Section~\ref{sec:relaxation} explains the connection between these the length spectrum and non-backtracking cycles, and how we use this connection in our derivation of the Truncated Non-backtracking Spectral Distance (TNBSD). Section~\ref{sec:nbm} describes our scalable algorithm for computing the non-backtracking matrix, as well as some of its spectral properties that help in the interpretation the TNBSD in terms of hubs and triangles. Section~\ref{sec:experiments} discusses the practical details of computing TNBSD as well as several case studies and applications. We conclude in Section \ref{sec:conclusions} with a discussion of limitations and future work.

\section{Theoretical Background}\label{sec:background}
Here we introduce two different theoretical constructions that may at first seem unrelated: the length spectrum of a metric space and the set of non-backtracking cycles of a graph. Our analysis pivots on the fact that the latter is a particular subset of the (domain of) the former.

\subsection{Length Spectrum}
Consider a metric space $X$ and a point $p \in X$. A closed curve that goes through $p$ is called a \emph{loop}, and $p$ is called the \emph{basepoint}. Two loops are \emph{homotopy equivalent} to one another relative to the basepoint when there exists a continuous transformation from one to the other that leaves the basepoint constant. The \emph{fundamental group} of $X$ with basepoint $p$ is denoted by $\pi_1(X, p)$ and is defined as the first homotopy group of $X$, i.e., the set of all loops in $X$ that go through $p$, modulo homotopy. Closed curves without a distinguished basepoint are called \emph{free loops}, and they correspond to conjugacy classes of $\pi_1(X, p)$. A well-known fact of homotopy theory is that if $X$ is path-connected then $\pi_1(X, p)$ is unique up to isomorphism, regardless of basepoint $p$. In the present work we only consider connected graphs, hence, we just write $\pi_1(X)$ when there is no ambiguity. For more on homotopy, refer to \cite{munkres2000topology,Hatcher}.

In general, the length spectrum is a function $\mathcal{L}$ from $\pi_1(X)$ of an arbitrary metric space $X$ to the real line, ${\mathcal{L}:\pi_1(X) \to \mathbb{R}}$, which assigns to each homotopy class of loops the infimum length among all of the representatives in its conjugacy class\footnote{The definition presented here is also known as \textit{marked} length spectrum. An alternative definition of the (unmarked) length spectrum does not depend on $\pi_1$; see for example \cite{leininger2007length}.}. Note, importantly, that the definition of length of a homotopy class considers the length of those loops not only in the homotopy class itself, but in all other conjugate classes. In the case of compact geodesic spaces, such as finite metric graphs which we consider in this work, this infimum is always achieved. For a finite graph where each edge has length one, the value of $\mathcal{L}$ on a homotopy class then equals the number of edges contained in the optimal representative. That is, for a graph $G=(V,E)$, $v \in V$, if $[c] \in \pi_1(G,v)$ and $c$ achieves the minimum length $k$ in all classes conjugate to $[c]$, we define $\mathcal{L}([c]) = k$.

Our interest in the length spectrum is supported by the two following facts. First, graphs are \emph{aspherical}. More precisely, once we are using a geometric realization of a graph $G$, its underlying topological space $\bar{G}$ is aspherical---all of its homotopy groups of dimension greater than 1 are trivial (see, for example, \cite{Hatcher}).\footnote{This follows from $G$ being homotopy equivalent to a bouquet of $k$ circles, where $k$ is the rank of the fundamental group of $G$. The universal covering of a bouquet of circles is contractible, which is equivalent to the space being aspherical.} Therefore, we study the only non-trivial homotopy group, the fundamental group $\pi_1(G)$. Second, Constantine and Lafont \cite{ConstantineLafont} showed that the length spectrum of a graph determines (a certain subset of) it up to isomorphism. Thus, we aim to determine when two graphs are close to each other by studying their length spectra relying on the main theorem of \cite{ConstantineLafont} on the marked length spectrum of spaces of dimension one. For completeness, we briefly mention the main result of \cite{ConstantineLafont}, which is known as {\it marked length spectrum rigidity}. For a metric space $X$ they define a subset $Conv(X)$, the minimal set to which $X$ retracts by deformation. Let $X_1,X_2$ be a pair of compact, non-contractible, geodesic spaces of topological dimension one. Their main theorem shows that if the marked length spectra of $X_1,X_2$ are the same, then $Conv(X_1)$ is isometric to $Conv(X_2)$. Now, when $X_1,X_2$ are graphs, as is our case, $Conv(X_i), i=1,2$, corresponds to the subgraph resulting from iteratively removing nodes of degree 1 from $G_i$; that is, $Conv(G_i)$ is the 2-core of $G_i$ \cite{batagelj2003m}. Thus, their main theorem states that \emph{when two graphs have the same length spectrum, their 2-cores are isomorphic}.

Given these results, it is natural to use the length spectrum as the basis of a measure of graph distance. Concretely, given two graphs, we aim to efficiently quantify how far their 2-cores are from being isomorphic by measuring the distance between their length spectra. In the next section, we explain our approach at implementing a computationally feasible solution for this problem.

\subsection{Non-Backtracking Cycles}\label{sub:nbc}
Here we introduce the non-backtracking cycles of a graph, and the associated non-backtracking matrix, and point out the connection between these and the theory of length spectra.

Let us set up some notation. Consider an undirected, unweighted graph $G=(V,E)$. For $e=(u,v)\in E$, define $e^{-1}$ as the same edge traversed in the inverse order, $e^{-1}=(v,u)$. A \emph{cycle} in $G$ is a sequence of edges $e_{1}e_{2}...e_{k}$ such that if $e_{i}=(u_{i},v_{i})$ then $v_{i}=u_{i+1}$ for $i=1,...,k-1$ and $v_{k}=u_{1}$. Here, $k$ is called \emph{length} of the cycle. A \emph{non-backtracking cycle} (NBC) is one where $e_{i+1}\neq e_{i}^{-1},\:i=1,...,k-1$ and $e_{k}\neq e_{1}^{-1}$; that is, an edge is never followed by its own inverse. Now let $|E|=m$. The associated non-backtracking matrix $B$ is the $2m\times2m$ matrix where each edge is represented by two rows and two columns, one per orientation: $(u,v)$ and $(v,u)$. For two edges $(u,v)$ and $(k,l)$, $B$ is given by 

\begin{equation}\label{eqn:nbm}
B_{k\to l, u\to v}=\delta_{vk}(1-\delta_{ul}),
\end{equation}

where $\delta_{ij}$ is the Kronecker delta. Thus, there is a 1 in the entry indexed by row $(k,l)$ and column $(u,v)$ when $u \neq l$ and $v = k$, and a 0 otherwise. Intuitively, one can interpret the $B$ matrix as the (unnormalized) transition matrix of a random walker that does not perform backtracks: the entry at row $(k,l)$ and column $(u,v)$ is positive if and only if a walker can move from node $u$ to node $v$ (which equals node $k$) and then to $l$, without going back to $u$.

The reason why NBCs are topologically relevant is, in a nutshell, because backtracking edges are homotopically trivial \cite{terras2010zeta}. Observe that the matrix $B$ tracks each pair of incident edges that do not comprise a backtrack; indeed, $tr(B^k)$ equals the number of NBCs of length $k$ in the graph. This fact will be fundamental in our later exposition. Observe too that $B$ is not symmetric, and hence its eigenvalues are in general complex numbers.

If one is interested not in $B$ itself, but rather in its eigenvalues, one may use the so-called Ihara determinant formula \cite{hashimoto1989zeta,bass1992ihara}, which says that the eigenvalues of $B$ different than $\pm 1$ are also the eigenvalues of the $2n \times 2n$ block matrix
\begin{equation}
B' = 
\begin{pmatrix}
  A & I - D \\
  I & 0
\end{pmatrix}
\end{equation}
where $A$ is the adjacency matrix, $D$ is the diagonal matrix with the degrees, and $I$ is the identity matrix of the appropriate size.

\section{Truncated Non-Backtracking Spectral Distance (TNBSD)}\label{sec:relaxation}
We want to quantify the distance (dissimilarity) between two graphs by measuring the distance between their length spectra. However, there are two main obstacles to such a task: i) computing the length spectrum of a given graph is not a straightforward task as it depends on the fundamental group, whose computation is prohibitive\footnote{More precisely, computing the fundamental group of a graph is trivial since it is a free group. However, what is prohibitive is describing the length spectrum as defined on the fundamental group since, to the best of the authors' knowledge, this would require the individual enumeration of the length of each of (its infinitely many) elements.}, and ii) it is not clear how to compare two length spectra functions that come from two distinct graphs since they are defined on disjoint domains (the fundamental groups of two distinct graphs)\footnote{In \cite{ConstantineLafont}, the authors need an isomorphism between the fundamental group of the spaces that are being compared --which is also computationally prohibitive.}. In order to overcome these obstacles, we propose an relaxed version of the length spectrum, which we denote by $\mathcal{L}'$ and whose construction comes in the form of a two-step aggregation of the values of $\mathcal{L}$; see Figure~\ref{fig:relaxed} for an overview of this procedure.

\begin{figure}[h]
\centering
\includegraphics[width=0.9\textwidth]{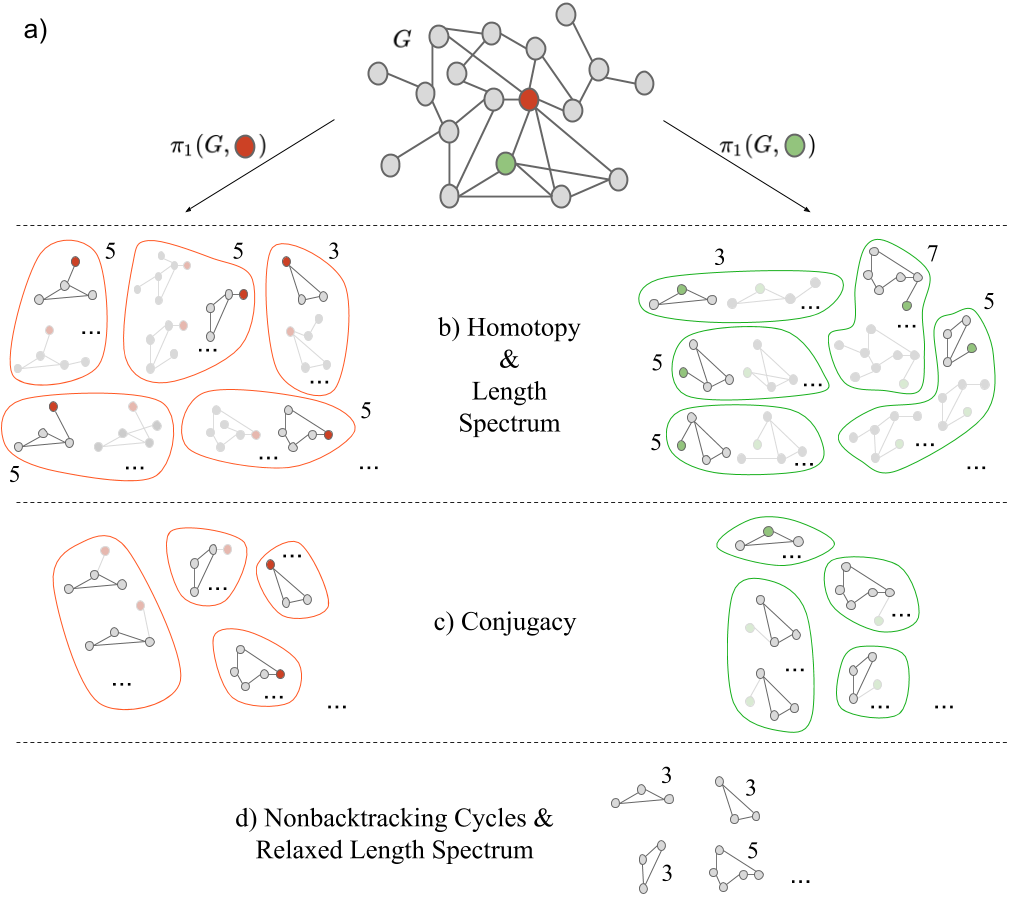}
\caption{\label{fig:relaxed}
Aggregating the values of the length spectrum. \textbf{a)} A graph $G$ with two nodes highlighted in red and green. These two nodes are used as basepoints to construct two versions of the fundamental group. \textbf{b)} The set of all cycles based at the red node (left) and green node (right). For either set of cycles, we encircle together those that are homotopy equivalent, thus forming a homotopy class. We highlight the representative with minimal length. Note that the lengths of corresponding cycles can change when the basepoints change. \textbf{c)} We have kept only the highlighted representative in each class in b) and encircled together those that are conjugate. In each conjugacy class, we highlight the (part of) each cycle that corresponds to the free homotopy loop. \textbf{d)} By taking one representative of each conjugacy class, and ignoring basepoints, we arrive at the free homotopy classes, or equivalently, at the set of non-backtracking cycles. Observe that the non-backtracking cycle at the top left, a triangle, originates from a homotopy class whose length is $5$ when the basepoint is the red node, but $3$ when the basepoint is the green node. The ellipses inside the closed curves mean that there are infinitely many more elements in each set. The ellipses outside the curves mean that there are infinitely many more classes or cycles.}
\end{figure}

\subsection{Relaxed Length Spectrum}
The first step of this procedure is to focus on the image, rather than the domain of the length spectrum (i.e., focus on the collection of lengths of cycles) in a way that will be clear in the next few lines. The second step is to aggregate these values by considering the sizes of the level sets of either length spectrum.

Concretely, when comparing two graphs $G, H$, instead of comparing $\mathcal{L}_G$ and $\mathcal{L}_H$ directly, we compare the number of cycles in $G$ of length 3 vs. the number of cycles in $H$ of the same length, as well as the number of cycles of length 4, of length 5, etc, thereby essentially considering the length spectra as histograms rather than functions. Theoretically, focusing on the size of the level sets provides a common ground to compare the two functions. In practice, this aggregation allows us to reduce the amount of memory needed to store either length spectrum because we no longer keep track of the exact composition of each of the infinitely many (free) homotopy classes. Instead, we only keep track of the frequency of their lengths. According to this aggregation, we define the \textit{relaxed} version of the length spectrum as the set of points $\mathcal{L}' = \{(k, n(k)):k=1,2,..\}$, where $n(k)$ is the number of conjugacy classes of $\pi_1$ (i.e., free homotopy classes) of length $k$.

The major downside of removing focus from the underlying group structure and shifting it towards (the histogram of values in) the image is that we lose information about the combinatorial composition of each cycle. Concretely, $\pi_1(G)$ holds information about the number of cycles of a certain length $k$ in $G$; this information is also stored in $\mathcal{L}'$. However, the group structure of $\pi_1(G)$ also allows us to know how many of those cycles of length $k$ are formed by the concatenation of two (three, four, etc.) cycles of different lengths. This information is lost when considering only the sizes of level sets of the image, i.e., when considering $\mathcal{L}'$. Fortunately, our experiments indicate that $\mathcal{L}'$ contains enough discriminatory information to distinguish between real and synthetic graphs effectively; see Section~\ref{sec:experiments}.

The next step makes use of the non-backtracking cycles (NBCs). We rely on NBCs because it is known (see e.g., \cite{hashimoto1989zeta,terras2010zeta}) that the set of conjugacy classes of $\pi_1(G)$ is in bijection with the set of NBCs of $G$. In other words, to compute the set $\mathcal{L}'$ we need only account for the lengths of all NBCs. Indeed, consider the non-backtracking matrix $B$ of $G$ and recall that $tr(B^k)$ equals the number of NBCs of length $k$ in the graph. This gives us precisely the set $\mathcal{L}' = \{(k, tr(B^k))\}_{k=1}^\infty$. Observe further that $tr(B^k)$ equals the sum of all of $B$'s eigenvalues raised to the $k$-th power. Therefore, the eigenvalues of $B$ contain all the information necessary to compute and compare $\mathcal{L}'$. In this way, we can study the (eigenvalue) spectrum of $B$, as a proxy for the (length) spectrum of $\pi_{1}$. Note that the use of $\mathcal{L}'$ presents one possible solution to the problems of how to compute and how to compare the length spectrum. We leave the investigation of alternative solutions to future lines of research.

\subsection{Properties of TNBSD}
The previous discussion yields a relaxed version of the length spectrum, $\mathcal{L}'$, that can be found efficiently: simply compute the associated matrix $B$ and its eigenvalues. We are finally prepared to state our definition of graph distance $d$ based on the length spectrum $\mathcal{L}$.

\begin{definition}\label{def:tnbsd}
Consider two graphs $G, H$, and write $\lambda_{k} = a_k + i b_k \in \mathbb{C}$ for the eigenvalues of the non-backtracking matrix of $G$ and $\mu_k = \alpha_k + i \beta_k$ for those of $H$, for $k=1,2,..,r$, where $r$ is some positive integer. Sort the eigenvalues in decreasing order of magnitude, $|\lambda_1| \geq |\lambda_{2}| \geq ... \geq |\lambda_{r}|$, $|\mu_1| \geq |\mu_{2}| \geq ... \geq |\mu_{r}|$. We define the \emph{truncated non-backtracking spectral distance (TNBSD)} between $G$ and $H$ as follows,
\begin{equation}
d(G, H) = d(\mathcal{L}'_{G}, \mathcal{L}'_{H}) = \sqrt{\sum_{k=1}^r | a_k - \alpha_k |^2 + |b_k - \beta_k|^2}
\end{equation}
\end{definition}

\begin{remark*}
Note that Definition~\ref{def:tnbsd} is the Euclidean distance between two $2r$-dimensional vectors whose entries are the real and imaginary parts of the eigenvalues of the respective non-backtracking matrices. The reason to separate the real and imaginary parts is that they have different interpretations with respect to features of complex networks such as hubs and triangles (see Sec.~\ref{sub:spectral}).
\end{remark*}

\begin{proposition}
$d$ is a pseudometric.
\end{proposition}

\begin{proof}
The function $d$ inherits from the Euclidean distance in its definition several desirable properties: non-negativity, symmetry, and, importantly, the triangle inequality. However, the distance between two distinct graphs may be zero when they share all of their eigenvalues. Thus, $d$ is not a metric over the space of graphs but a pseudometric.
\end{proof}

The authors of \cite{koutra2013deltacon} propose a few axioms and properties that a measure of graph similarity should satisfy. Here, we present the equivalent axioms and properties for a measure of graph dissimilarity (distance) and show that the TNBSD satisfies them. The axioms are as follows:

\begin{enumerate}
\item[A1.] Identity: $d(G,G) = 0$.
\item[A2.] Symmetry: $d(G,H) = d(H,G)$.
\item[A3.] Divergence: $d(K_n, \bar{K}_n) \to \infty$ as $n\to \infty$, where $K_n$ is the complete graph and $\bar{K}_n$ is the empty graph (a graph with zero edges).
\end{enumerate}

\begin{proposition}
$d$ satisfies axioms A1-A3.
\end{proposition}

\begin{proof}
Axioms A1 and A2 are satisfied because $d$ is a pseudometric. Axiom A3 is satisfied by observing that the non-backtracking matrix of the empty graph has zero rows, and thus it has no eigenvalues, while the eigenvalues of the complete graph grow as the number of nodes grows. Thus we may accept that $d$ satisfies axiom A3 by convention. If the reader is not satisfied by the fulfillment of an axiom by mere convention, we offer an alternative. We may compare the complete graph $K_n$ to the graph on $n$ nodes with a single edge linking two arbitrary nodes (an \emph{almost} empty graph), in which case its non-backtracking matrix has two rows and two eigenvalues equal to zero. Axiom A3 is still satisfied.
\end{proof}

%P: We could mention in the proof above as a concrete example that two graphs that differ only by the addition of a tree have the same NBCs.

%\begin{proposition}
%\draft{$d$ satisfies properties ?-?-? from DELTACON.}
%\end{proposition}

% P: there is also an interesting property of the eigenvalues of B that we could appeal to here, the bound coming from the existence of a complete subgraph...
%L: Yes, I'm all for adding our own axioms and providing an intuition as to why we think they are necessary/desirable

\subsection{Using \texorpdfstring{$\mathcal{L}'$}{L'} instead of 
\texorpdfstring{$\mathcal{L}$}{L}}

Although the truncated non-backtracking spectral distance satisfies all desired axioms and properties, we have deviated from the original definition of the length spectrum in important ways. In fact, as pointed out earlier, $\mathcal{L}'$ is admittedly weaker than $\mathcal{L}$ and thus one may ask if there are theoretical guarantees that the relaxed version of the length spectrum will keep some of the discriminatory power of the original. Indeed, even though the main inspiration for our work is the main result of \cite{ConstantineLafont}, we can still trust the eigenvalue spectrum of $B$ to be useful when comparing graphs. On the one hand, the spectrum of $B$ has been found to yield fewer isospectral graph pairs when compared to the adjacency and Laplacian matrices in the case of small graphs \cite{durfee2015distinguishing}. On the other hand, $B$ is tightly related to the theory of graph zeta functions \cite{hashimoto1989zeta}, in particular the Ihara Zeta function, which is known to determine several graph properties such as girth, number of spanning trees, whether the graph is bipartite, a forest, or regular, among others \cite{cooper2009properties}. Thus, both as a relaxed version of the original length spectrum, but also as an object of interest in itself, we trust the eigenvalue spectrum of the non-backtracking matrix $B$ to be of use when determining the dissimilarity between two graphs.

\section{Non-Backtracking Matrix: Algorithm and Properties}\label{sec:nbm}
For the rest of this work, we focus on the nonbactracking matrix and its properties. We now present a spectral analysis which will aid in the study of several aspects of the proposed distance TNBSD. We present an algorithm for computing $B$, as well as describe properties of the eigenvalue distribution in connection with features of complex networks.

\subsection{Computing \texorpdfstring{$B$}{B}}
Given a graph with $n$ nodes and $m$ undirected edges, define the $n \times 2m$ incidence matrices $P_{x,u \to v} = \delta_{xu}$ and $Q_{x,u \to v} = \delta_{xv}$, and write $C = P^T Q$. Observe that $C_{k\to l, u\to v} = \delta_{vk}$. Therefore,
\begin{equation}
B_{k\to l, u\to v} = C_{k\to l, u\to v} (1 - C_{u\to v, k\to l})
\end{equation}
Note that an entry of $B$ may be positive only when the corresponding entry of $C$ is positive. Therefore, we can compute $B$ in a single iteration over the nonzero entries of $C$. Now, $C$ has a positive entry for each pair of incident edges in the graph, thus we find $nnz(C) = O(n\langle k^{2}\rangle)$, where $\langle k^{2}\rangle$ is the second moment of the degree distribution, and $nnz(C)$ is the number of non-zero entries in $C$. Since computing $P, Q$ takes $O(m)$ time, we can compute $B$ in time $O(m + n\langle k^{2}\rangle)$. For example, in the case of a power-law degree distribution with exponent ${2 \leq \gamma \leq 3}$, the runtime of our algorithm falls between $O(m+n)$ and $O(m+n^2)$. Note that if a graph is given in adjacency list format, one can build $B$ directly from the adjacency list in time $\Theta(n\langle k^2 \rangle - n\langle k \rangle)$ by generating a sparse matrix with the appropriate entries set to $1$ in a single iteration over the adjacency list.

\subsection{Spectral Properties}\label{sub:spectral}
Observe that the sparsity of $B$ grows with the second moment of the degree distribution.

\begin{lemma}
Consider the non-backtracking matrix $B$ of a graph $G$ with $n$ nodes and let $nnz(B)$ be the number of non-zero elements therein. Then,
\begin{equation}
nnz(B) = n\big(\langle k^2 \rangle - \langle k \rangle\big),
\end{equation}
where $\langle k \rangle$ and $\langle k^2 \rangle$ are the first and second moments of the degree distribution of $G$, respectively.
\end{lemma}

\begin{proof}
This is seen by using Equation~\ref{eqn:nbm} to sum over all the elements of $B$.
\end{proof}

Contrast this to ${nnz(A) = n \langle k \rangle}$, where $A$ is the adjacency matrix of the graph. Experimentally, we have found that the larger $\langle k^2 \rangle$, the larger the variance of $B$'s complex eigenvalues along the imaginary axis (Figure~\ref{fig:random_eigenvalues}).

Next, we turn to $B$'s eigenvalues and their relation to the number of triangles. Write $\lambda_k = a_k + i b_k \in \mathbb{C}$ for the eigenvalues of $B$, $k=1,2,..,2m$. The number of triangles in a network is proportional to $tr(B^{3}) = \sum_k Re(\lambda_k^3)$,\footnote{The imaginary part of this expression vanishes because the complex eigenvalues of a matrix always come in conjugated pairs.} which, by a direct application of the binomial theorem, equals
\begin{equation}\label{eqn:triangles}
tr(B^3) = \sum_{k=1}^{2m} a_{k}(a_k^2 - 3 b_{k}^{2}).
\end{equation}
On the one hand, $B$'s eigenvalues tend to fall on a circle in the complex plane \cite{krzakala2013spectral,angel2007non,wood2017limiting,bordenave2015non}. On the other hand, if $\sum_k a_k^2$ is large and $\sum_k b_k^2$ is small (implying a large number of triangles), the $\lambda_k$ cannot all fall too close to the circle. Hence, the more triangles in the graph, the less marked the circular shape of the eigenvalues (Figure~\ref{fig:random_eigenvalues}).

Finally, a note of practical importance on the spectrum of $B$. The multiplicity of the eigenvalue $0$ equals the number of edges outside of the 2-core of the graph. For example, a tree, whose 2-core is empty, has all its eigenvalues equal to $0$. On the one hand, we could use this valuable information as part of our method to compare two graphs. On the other hand, the existence of zero eigenvalues does not change the value of $tr(B^k),k\geq0$, and thus leaves $\mathcal{L}'$, the relaxed length spectrum, intact. Moreover, removing the nodes of degree one reduces the size of $B$ (or the sparsity of $B'$, see Sec.~\ref{sub:nbc}), which makes the computation of non-zero eigenvalues faster.

\section{Experiments}\label{sec:experiments}
We discuss practical aspects of computing the truncated non-backtracking spectral distance (TNBSD), as well as explain how to fine tune it to be more sensitive to triangles and degree distribution. We also present experimental evidence of its discriminatory power when comparing random and real graphs.

\subsection{Computing \texorpdfstring{$d$}{d}}
Given two graphs $G,H$ and a positive integer $r$, we compute the distance between the two graphs in three steps; see Algorithm \ref{alg:distance}. First, remove all nodes of degree one from either graph. As mentioned previously, nodes of degree one do not affect the spectrum $\mathcal{L}'$, and their removal makes the computation faster. Note that after removing a node of degree one, another node's degree might decrease from $2$ to $1$. Thus, we need to iterate this removal until all nodes in the graph have degree at least $2$. (This process is called ``shaving'' in graph mining, and yields the 2-core of the graph.) Second, compute the block matrix $B'$ (Sec.~\ref{sub:nbc}) from either graph and compute its largest $r$ eigenvalues. Third, write these as $\lambda_k=a_{k}+ib_{k}$ for $G$ and $\mu_k=\alpha_{k}+i \beta_{k}$ for $H$, where $|\lambda_k| \geq |\lambda_{k+1}|$ and $|\mu_k| \geq |\mu_{k+1}|$ for $k=1,...,r-1$, and assign to $G$ the feature vector $v_{1}=(\mathbf{a}, \mathbf{b})=(a_{1},a_{2},...,a_{r},b_{1},b_{2},...,b_{r})$, and to $H$ assign $v_{2}=(\bm{\alpha}, \bm{\beta})=(\alpha_{1},\alpha_{2},...,\alpha_{r},\beta_{1},\beta_{2},...,\beta_{r})$. Finally, compute the distance between $G$ and $H$ as $\|v_{1}-v_{2}\|$, where $\|\cdot\|$ is the Euclidean norm. See Figure~\ref{fig:distance} (left column) for results of applying this distance measure to random graph models.

\begin{algorithm}
\caption{Truncated non-backtracking Spectral Distance}
\label{alg:distance}
\hspace*{\algorithmicindent} \textbf{Input: Graphs $G,H$, positive integer $r$} \\
\hspace*{\algorithmicindent} \textbf{Output: real number $d$, distance between $G, H$}
\begin{algorithmic}[1]
\STATE{$\tilde{G},\tilde{H} \leftarrow \operatorname{shave}(G), \operatorname{shave}(H)$}
\STATE{$\{\lambda_{k}\}_{k=1}^r, \{\mu_{k}\}_{k=1}^r \leftarrow$ largest eigenvalues of $B'$ corresponding to $\tilde{G}, \tilde{H}$}
\STATE{$v_1, v_2 \leftarrow (a_{1}, .., a_{r}, b_{1}, .., b_{r}), (\alpha_{1}, .., \alpha_{r}, \beta_{1}, .., \beta_{r})$, \\
for ${\lambda_{k}=a_{k} + i b_{k}}$ and ${\mu_{k}=\alpha_{k} + i \beta_{k}}$, with $k=1,...,r$}
\STATE{$d \leftarrow \|v_1 - v_2\|$}
\RETURN $d$
\end{algorithmic}
\end{algorithm}

\subsection{Fine-tuning}\label{sub:tuning}
One advantage of this distance is that it can be fine tuned to capture certain features, namely those mentioned in Sec.~\ref{sub:spectral}. For instance, if number of triangles is of particular interest, one may accentuate the effect of equation \ref{eqn:triangles} as follows. If one increases the sum of squares of the real parts and decreases the sum of squares of imaginary parts, one would be artificially increasing the number of triangles in the graph. Hence, to emphasize this, one may compute the distance using the modified feature vectors

\begin{equation}\label{eqn:sigma}
v_1' = (\sigma \mathbf{a}, \mathbf{b}/\sigma), v_2' = (\sigma \bm{\alpha}, \bm{\beta}/\sigma),
\end{equation}

for some real number $\sigma\geq1$. We have also observed experimentally that the spread of the imaginary parts of the eigenvalues increases as the second moment of the degree distribution increases. Hence, if degree distribution is of interest, one may emphasize this effect by using instead the feature vectors

\begin{equation}
\label{eqn:eta}
\begin{aligned}[b]
v_1' &= (|\lambda_1|^\eta a_{1},...,|\lambda_r|^\eta a_{r},|\lambda_1|^\eta b_{1},...,|\lambda_r|^\eta b_{r}) \nonumber \\
v_2' &= (|\lambda_1|^\eta \alpha_{1},...,|\lambda_r|^\eta \alpha_{r},|\lambda_1|^\eta \beta_{1},...,|\lambda_r|^\eta \beta_{r})
\end{aligned}
\end{equation}

with $\eta \in \mathbb{R}$ and $\eta > 0$. See Fig.~\ref{fig:distance} for an example of using these modifications when comparing the random graphs shown in Fig.~\ref{fig:random_eigenvalues}.

\begin{figure}[h]
\centering
\includegraphics[width=0.9\textwidth]{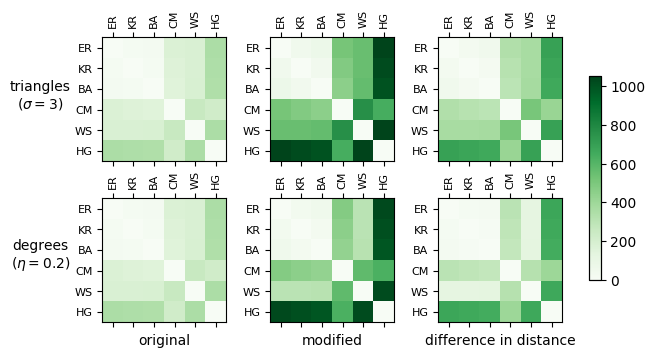}
\caption{\label{fig:distance}
(Best viewed in color.) Fine-tuning TNBSD on various random graphs. The graphs here are the same as the ones described in Fig.~\ref{fig:random_eigenvalues}.  Thus, ER represents Erd\"os-R\'enyi, BA represents Barab\'asi-Albert (BA), KR represents Stochastic Kronecker Graphs, CM represents Configuration Model with power law degree distribution ($\gamma=2.3$), WS represents for Watts-Strogatz, and HG represents Hyperbolic Graphs.  Left column: original (unmodified) TNBSD between the average eigenvalue vectors of the graphs. Middle column: modified TNBSD fine-tuned to triangles (top row) and degree distribution (bottom row). Right column: difference between original and modified TNBSD values. Observe that after fine-tuning to triangles, the distance to HG is increased the most since HG has by far the most triangles across all random graph models used. Similarly, when fine-tuning for degrees, both CM and HG are emphasized since they have-heavy tailed degree distributions. (At this number of nodes, $n=5\times 10^4$, the degree distribution of BA is not as heavy-tailed.)} 
\end{figure}

\subsection{Case Study 1: Clustering Random Graphs}\label{sub:cluster}
In the first case study, we compute the truncated non-backtracking spectral distance (TNBSD) between random graphs generated with different random graph models in order to find clusters corresponding to said models. We use a Gaussian mixture which we optimize with the Expectation Maximization (EM) algorithm (\cite{Murphy:2012:MLP:2380985}, Ch. 11). Since the Gaussian probability density function assigns likelihood based on the distance from an arbitrary point to the mean of the distribution, this setup explicitly uses the TNBSD to perform the learning task. In this case study, our purpose is to showcase the effectiveness of the TNBSD, as well as the fine-tuning mechanisms presented in a previous section, in an unsupervised learning setting. Our purpose is not to perform an exhaustive sweep of parameter space.

The experimental setup is as follows. We generate $50$ graphs of each of six different random graph models, for a total of $300$ graphs. Each graph has $5\times10^4$ nodes and approximate average degree $\langle k \rangle = 15$ (see Fig.~\ref{fig:random_eigenvalues}). The random graph models used were Erd\"os-R\'enyi (ER) \cite{erdos1960p,bollobas2001random}, Stochastic Kronecker Graph (KR) \cite{leskovec2010kronecker,seshadhri2013depth}, Barab\'asi-Albert (BA) \cite{barabasi1999emergence}, Configuration Model with power law degree distribution with exponent $\gamma=2.3$ (CM) \cite{newman2003structure}, Watts-Strogatz (WS) \cite{watts1998collective}, and Hyperbolic Graph with degree distribution exponent $\gamma=2.3$ (HG) \cite{krioukov2010hyperbolic,aldecoa2015hyperbolic}. We compute the largest $r=200$ eigenvalues of each graph. For each graph $j=1,..,300$, we generate the vector $v_j=(a_1, ..,a_r, b_1, .., b_r)$, where $\lambda_k = a_k + ib_k$, $k=1,...,r$ are the eigenvalues of the non-backtracking matrix. We use Kernel Principal Component Analysis (\cite{Murphy:2012:MLP:2380985}, Ch. 14) on the set of vectors $\{v_j\}$ to reduce the number of dimensions of the data set to two, for visualization purposes; the kernel used was cosine similarity. Next, we employ the EM algorithm to estimate data density in 2D space and predict which Gaussian component each graph is most likely to have come from (Fig.~\ref{fig:em}).

Using the unmodified distance, the results are as follows: three clusters are easily discernible (CM, HA, WS), while the other three (BA, ER, KR) are not quite so well defined (Figure~\ref{fig:em}a). However, as explained in Section~\ref{sub:tuning}, we can use the interpretable geometric features of the eigenvalue distribution to improve this result. We know that ER and BA will differ greatly by their degree distribution; specifically, BA will have large second moment of the degree distribution, $\langle k^2 \rangle$, at large number of nodes. However, the number of nodes used here ($5\times 10^4$) may not be enough to show this feature. Therefore, we need to emphasize this feature and make the distance measure more sensitive to $\langle k^2 \rangle$ by using Equation~\ref{eqn:eta}. We find that a value of $\eta=0.6$ gives the desired result: the cluster corresponding to BA graphs is more easily discernible from ER, KR (Figure~\ref{fig:em}b). Furthermore, we know that KR and ER differ in the number of expected triangles. Thus, using Equation~\ref{eqn:sigma}, we find the parameter $\sigma=11$ that makes ER and KR graphs more distinguishable (Figure~\ref{fig:em}c). The combination of these two fine-tuned parameters allows us to recover with great accuracy the original random models originating the graphs (Figure~\ref{fig:em}d). The best accuracy achieved across all random initializations of the experiment was $98.66\%$.

\begin{figure}[t]
\centering
\includegraphics[width=1.\textwidth]{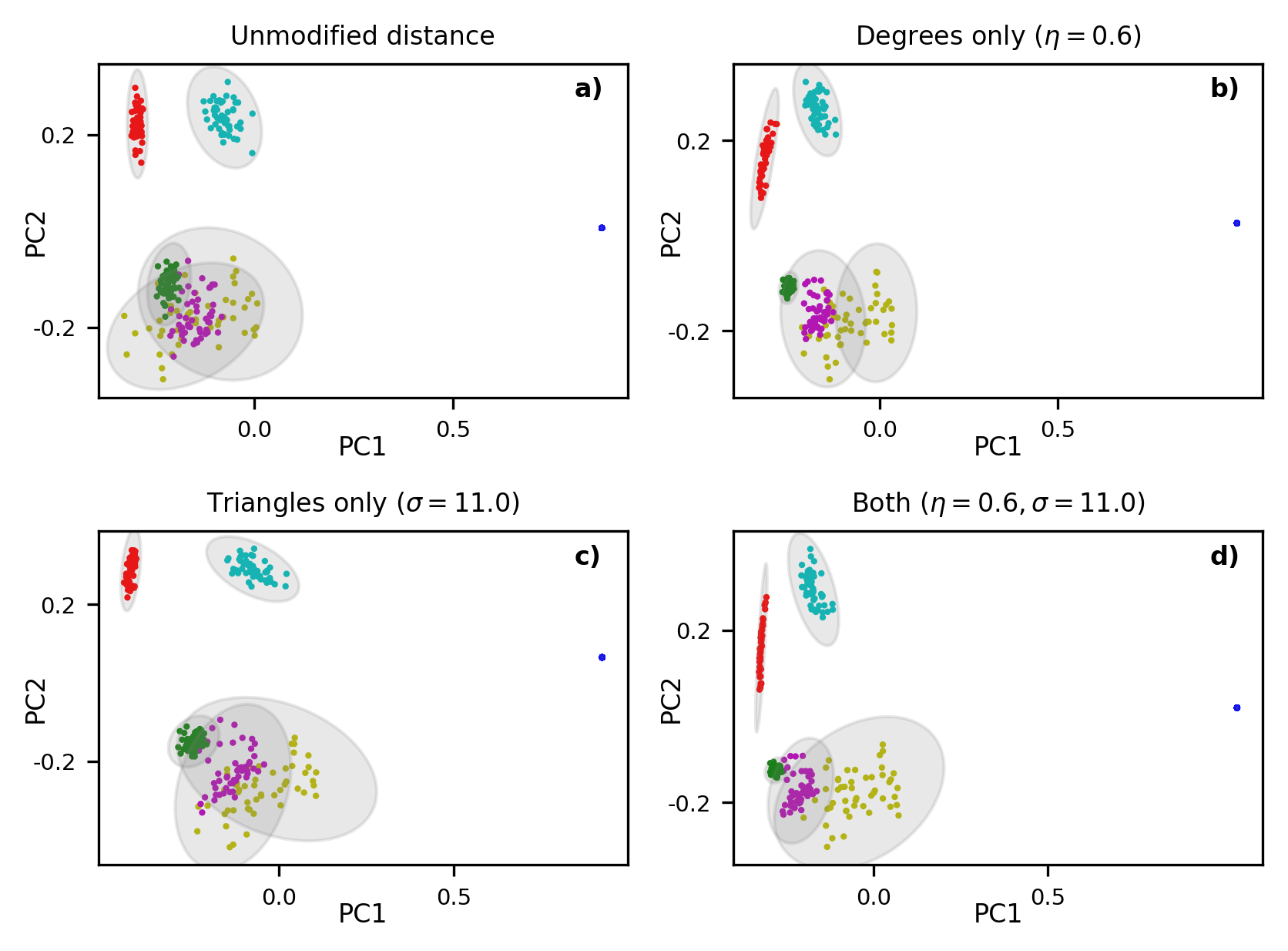}
\caption{\label{fig:em}
(Best viewed in color.) Using TNBSD to cluster random graphs. We compute the largest $r=200$ eigenvalues random graphs of six different models: Erd\"os-R\'enyi (yellow), Stochastic Kronecker Graph (magenta), Barab\'asi-Albert (green), Configuration Model with power law degree distribution and $\gamma=2.3$ (cyan), Watts-Strogatz (blue), and Hyperbolic Graphs with $\gamma=2.3$ (black). We generated 50 graphs per model. Visualized are the first two principal components of the data set after applying Kernel PCA with cosine similarity. The clusters are found with the Expectation Maximization algorithm optimizing a Gaussian mixture model. The ellipses (gray) are centered around the estimated clusters. We show results using the unmodified distance in (\textbf{a}),  modified to emphasize only degree distribution in (\textbf{b}), modified to emphasize only triangles in (\textbf{c}), and modified to emphasize both degree distribution and triangles in (\textbf{d}). All but four data points out of 300 are clustered correctly (accuracy 98.66\%) in \textbf{d}. All graphs have $n=5\times10^4$ nodes and average degree approximately $\langle k \rangle = 15$. See definitions of $\eta, \sigma$ in Sec.~\ref{sub:tuning}.}
\end{figure}

We wish to put this result in the context of other state-of-the-art graph distance methods. For example, the authors of \cite{bento2018family} claim their method is able to cluster certain random graphs models with no misclassifying errors when the sizes have $N=50$ nodes. The methods ORTHOP and ORTHFR in \cite{bento2018family} are a direct relaxation of the graph isomorphism problem based on the chemical distance, which tries to find a perfect node alignment between two graphs. Thus we expect them to perform quite well in this experiment.

In this context, we wish to study the performance of the TNBSD as the number of nodes varies. Can the TNBSD come close to ORTHOP and ORTHFR when clustering small graphs? More generally, how far is TNBSD from identifying the isomorphism class of (the 2-core of) a graph, which was the original promise of the theory of the length spectrum? For this purpose, we execute the same experimental setup as above but on increasingly smaller graphs and compare to an approximation of the graph isomorphism problem, namely ORTHOP and ORTHFR \cite{bento2018family}. See Fig.~\ref{fig:sizes} for results. TNBSD can achieve comparable performance to ORTHOP/ORTHFR at $N=5\times 10^4$, while achieving acceptable performance across all other graph sizes (when using fine-tuning parameters). We hypothesize that the drop in performance of TNBSD in smaller graphs is due to the fact that smaller graphs have fewer eigenvalues (each with smaller absolute value), which yields a less distinguishable pattern on the complex plane. However, when the graphs are larger, the increase in number of eigenvalues yields geometric patterns that are larger and can be distinguished (and fine-tuned) more easily.

\begin{figure}[t]
\centering
\includegraphics[scale=.5]{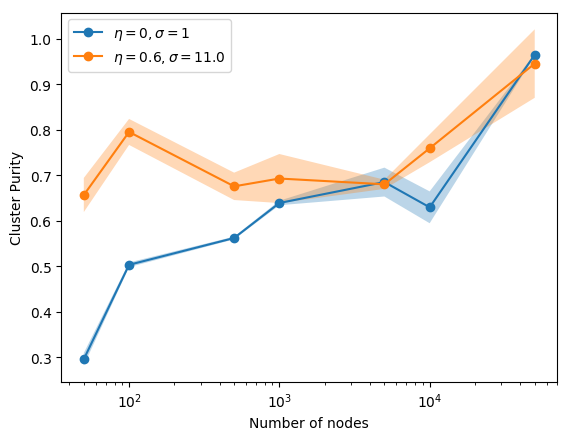}%[width=.8\textwidth]{sizes}
\caption{\label{fig:sizes}
Performance of random graph clustering as number of nodes varies, using unmodified TNBSD (blue) and fine-tuned with same parameters as in Fig.~\ref{fig:em}d (orange, ). Authors of \cite{bento2018family} claim their methods (ORTHOP and ORTHFR) achieve 100\% cluster purity at $N=50$ nodes (not pictured).  A clustering algorithm that classifies graphs purely at random would yield 16.66\% purity. Variance in purity is due to stochasticity of random graph models and random initializations of the clustering algorithm. $\eta$ fine tunes to degree distribution, $\sigma$ fine tunes number of triangles. See Sec.~\ref{sub:tuning} for definitions.}
\end{figure}

\subsection{Case Study 2: Dissimilar Samples of the Same Graph}\label{sub:similar}
In this case study we take several samples of the same real life network with different sampling algorithms, and measure the distance between them with the purpose of determining which samples were taken with the same algorithm. In doing so we also show the visualization capabilities afforded by the non-backtracking eigenvalues (Fig.~\ref{fig:samples}).

For this experiment, we use the web graph of web pages belonging to Stanford University \cite{DBLP:journals/im/LeskovecLDM09}. This graph has $n=281903$ nodes and $m=2312497$ edges. We took two samples with each of the following sampling algorithms: node sampling (NS), edge sampling (ES), random walk sampling (RW), random walk sampling with jump (RJ) \cite{DBLP:journals/tkdd/AhmedNK13}, for a total of eight samples. The samples were taken from random seeds until a minimum  of $5\%$ of existing edges were observed. Jump probability for RJ was $p=0.3$. After visualizing the non-backtracking eigenvalues of each sample graph, we observe there are regions of the complex plane that are consistently occupied by only two of these samples at the same time. However, visualization of the eigenvalues on the complex plane (Fig.~\ref{fig:samples}a) or in a reduced space (Fig.~\ref{fig:samples}b) does not yield definitive answers. Therefore, we proceed to apply a statistical test to determine which samples were taken with the same algorithm. We assume that the underlying original network determines a continuous probability distribution over the complex plane with support set $\mathcal{S}$, and that each sampling algorithm determines a distinct probability distribution over $\mathcal{S}$. We assume, further, that the eigenvalues of each sample network are independent observations drawn from these distributions. Hence, to answer the question of which of those samples are taken from the same distribution, we use the Kolmogorov-Smirnov test on each pair of two samples under the null hypothesis that the samples come from the same distribution. (Fig.~\ref{fig:samples}c) shows that this test is capable of determining which samples are taken from the same algorithm.

\begin{figure}[t]
\centering
\includegraphics[width=.9\textwidth]{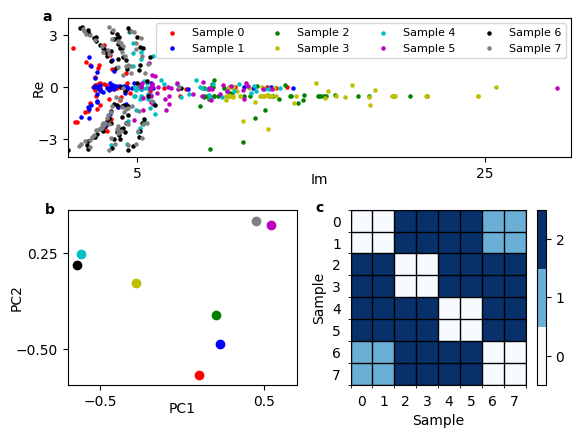}
\caption{\label{fig:samples}
(Best viewed in color.) Visualizing distinct samples of the same graph. Eight samples were taken of the same data set (see Sec.~\ref{sub:similar}). We plot the largest $r=200$ non-backtracking eigenvalues of each sample, one per color, \textbf{(a)}. Samples 2 and 3 (green and yellow) achieve similar behavior, as do Samples 0 and 1 (blue and red), Samples 6 and 7 (gray and black), and Samples 4 and 5 (cyan and magenta). We may thus posit that those are the pairs of samples that come from the same algorithm. However, when visualizing the principal components of each sample after applying Kernel PCA with cosine similarity, \textbf{(b)}, we do not get confirmation of this hypothesis. Hence, we rely on the Kolmogorov-Smirnov statistic. We assume that each sampling algorithm determines two probability distributions over the real numbers (one for the real axis and one for the imaginary axis). We test the hypothesis that each pair of samples comes from the same underlying distributions using the Kolmogorov-Smirnov test. In \textbf{(c)} we report the number of tests in which the null hypothesis is rejected for each pair of samples. Here we confirm that for those pairs of samples identified in \textbf{(a)}, the null cannot be rejected, while there is enough evidence to reject all other pairs. All tests performed at $90\%$ significance level and with Bonferroni correction for multiple comparisons of $m=14$ (each sample is compared to seven others twice --one for the distribution of real parts of the eigenvalues and one for the imaginary part).}
\end{figure}

\subsection{Case Study 3: Degree-Preserving rewiring}\label{sub:rewiring}
In this case study, we observe the performance of TNBSD in the presence of structural noise in the graph. The purpose is to elucidate the saliency of the structural properties detected by this distance measure, and to determine how robust they are when in the presence of noise.

The setup is as follows. We consider a graph $G$, and compute its non-backtracking eigenvalues. Then, generate an ensemble of random graphs with the configuration model that have the same degree distribution as $G$. We compute the average distance from $G$ to this ensemble. This distance represents the structural saliency that TNBSD is detecting. In other words, if $G$ is close to the random ensemble in terms of TNBSD, then its non-backtracking structure is not salient and this would be a counter-indication to the use of TNBSD. Moreover, we introduce structural noise to $G$ by performing degree-preserving randomization \cite{maslov2002specificity} on $G$, and measuring the distance between the rewired graph and $G$. By varying the rewiring parameter (i.e., the probability of rewiring an edge), we expect that the rewired versions of $G$ will move closer and closer to the random ensemble of configuration model graphs; thus, we expect the distance to increase from $0$ to the average distance to the random ensemble. See Fig.~\ref{fig:rewiring} for results on random graphs and the samples used in Case Study 2.

The results show that TNBSD is able to distinguish the original graph from noisy versions across a wide range of the rewiring parameter in several cases. However, both BA and KR are outliers since they are indistinguishable from other graphs with the same degree distribution, even with a small number of rewirings. This partly explains the observation in Case Study 1 that TNBSD was not able to satisfactorily distinguish between KR and BA graphs before fine-tuning. 

\begin{figure}[t]
\centering
\includegraphics[width=.9\textwidth]{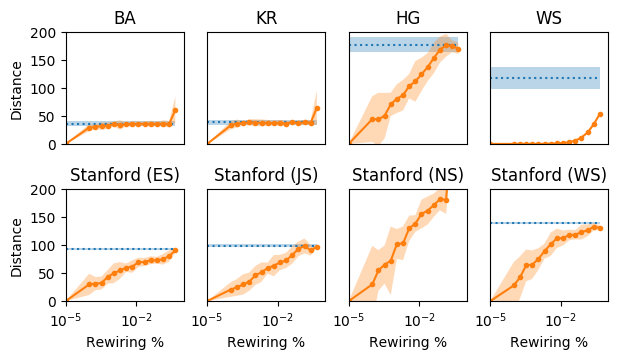}
\caption{\label{fig:rewiring}
Performance of TNBSD in the presence of structural noise (degree-preserving rewiring). Average distance between a graph and an ensemble of graphs with the same degree distribution (blue), and distance between the original graph and rewired versions, by percentage of rewired edges (orange). Shaded regions show two standard deviations around the mean. In most cases, TNBSD is able to distinguish between the original graph and noisy versions across several orders of magnitude of the rewiring parameter, in some cases even after 20\% of the edges have been rewired. BA and KR are indistinguishable from other graphs with the same degree distribution, which highlights the need to use fine-tuning in applications.}
\end{figure}

\subsection{Case Study 4: Enron data set}\label{sub:enron}
In this last case study we apply the TNBSD to the well-known Enron emails data set \cite{klimt2004introducing,agsm,guardian,NYT}. From it, we extract a who-emails-whom network, treat it as undirected and unweighted, and aggregate it both daily and weekly; see Fig.~\ref{fig:enron}. The purpose is to recover general common sense features of this data set, such as the periodicity of weekly communications, as well as perform anomaly detection: we expect to see anomalies in the distance measured between graphs of this data set whenever a major event in the Enron scandal occurred. We were able to recover both of these features (Fig.~\ref{fig:enron}).

\begin{figure}[t]
\centering
\includegraphics[width=.9\textwidth]{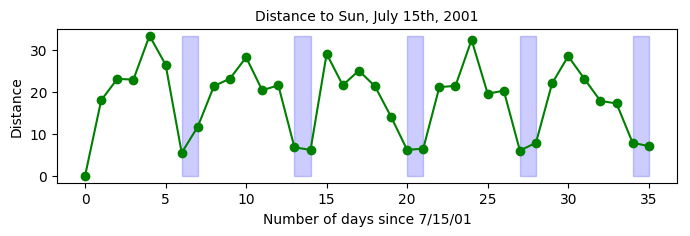}
\includegraphics[width=.9\textwidth]{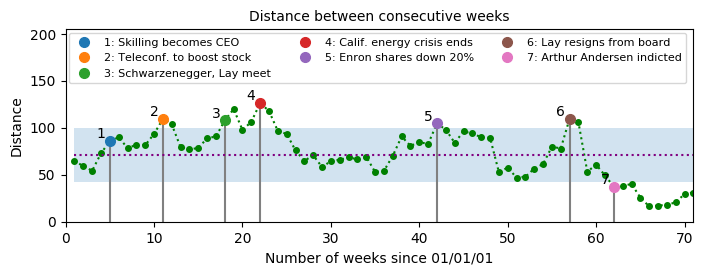}
\caption{\label{fig:enron}
(Best viewed in color.) Applying TNBSD to the Enron data set. \textbf{Top:} Data aggregated into daily graphs and compared to Sunday, July 15th, 2001. The periodicity of weekly communications is recovered; that is, graphs corresponding to Saturdays and Sundays are closer to each other than they are to weekdays. \textbf{Bottom:} Data aggregated into weekly graphs and compared to the previous week. We highlight the mean distance (dashed line) and one standard deviation around it (shaded area). Each week that falls outside of the shaded area coincides with a known event during the Enron scandal and subsequent collapse.}
\end{figure}

\section{Conclusions}\label{sec:conclusions}
In this work, we have focused on the problem of deriving a notion of graph distance for complex networks based on the length spectrum function. We add to the repertoire of distance measures \cite{soundarajan2014sdm,koutra2013deltacon,bagrow2018information,bento2018family,PhysRevE.86.036104,schieber2017quantification,chowdhury2017distances,chowdhury2018metric} the Truncated Non-Backtracking Spectral Distance (\textbf{TNBSD}): a principled, interpretable, efficient, and effective measure that takes advantage of the fact that the non-backtracking cycles of a graph can be interpreted as its free homotopy classes. \textbf{TNBSD} is principled because it is backed by the theory of the length spectrum, which characterizes the 2-core of a graph up to isomorphism; it is interpretable because we can study its behavior in the presence of structural features such as hubs and triangles, and we can use the resulting geometric features of the eigenvalue distribution to our advantage; it is efficient because it takes no more time than computing a few of the largest eigenvalues of the non-backtracking matrix; and we have presented extensive experimental evidence to show that it is effective at discriminating between complex networks in various contexts, including visualization, clustering, sampling, and anomaly detection.

\paragraph{Limitations} There are two major limitations of \textbf{TNBSD}. First, it relies on the assumption that the non-backtracking cycles contain enough information about the network. Concretely, the usefulness of the \textbf{TNBSD} will decay as the 2-core of the graph gets smaller. For example, trees have an empty 2-core, and all of its non-backtracking eigenvalues are equal to zero. In order to compare trees, and more generally, those nodes outside the 2-core of the graph, the authors of \cite{durfee2015distinguishing} propose several different strategies, for example adding a ``cone node" that connects to every other node in the graph. However, a broad class of complex networks will not look like trees.  The utility of \textbf{TNBSD} on this class of networks was extensively showcased in Sec.~\ref{sec:experiments}. Second, definition~\ref{def:tnbsd} is only one possible way to solve the problems of how to compute and how to compare the length spectrum function. One point of possible improvement is how we choose which eigenvalues are compared to which others. Currently, we sort the eigenvalues by magnitude in order to compare them, but this may not be the best setup for comparison, especially because there are usually many eigenvalues with approximately the same magnitude. Indeed, we have already hinted at a possible solution to this problem when we applied  \textbf{TNBSD}, not to the eigenvalues themselves, but to their projection on the space of principal components after performing Kernel PCA with cosine similarity (Sections~\ref{sub:cluster} and \ref{sub:similar}).

\paragraph{Future work} There are many other avenues to explore in relation to how to exploit the information stored in the length spectrum and the fundamental group. As mentioned in Sec.~\ref{sec:relaxation}, the major downside of the relaxed length spectrum $\mathcal{L}'$ is the fact that we lose information stored in the combinatorics of the fundamental group. That is, $\mathcal{L}'$ stores information of the frequency of lengths of free homotopy classes, but no information on their concatenation -- i.e., the group operation in $\pi_1(G)$. One way to encapsulate this information is by taking into account not only the frequency of each possible length of non-backtracking cycles, but also the number of non-backtracking cycles of fixed lengths $\ell_1$ and $\ell_2$ that can be concatenated to form a non-backtracking cycle of length $\ell_3$. It remains an open question how to compute this information efficiently using the non-backtracking matrix for all values of the parameters $\ell_1, \ell_2, \ell_3$, which range freely on the positive integers.

We conclude by mentioning that we hope this work paves the road for more research along the lines of topological and geometric data analysis of complex networks focusing on introducing and exploiting novel theoretical concepts such as the length spectrum function and the fundamental group.

\section*{Acknowledgments}
We thank Evimaria Terzi for her contributions to an earlier version of this work. Torres and Eliassi-Rad were supported by NSF CNS-1314603 and NSF IIS-1741197. Su\'arez-Serrato was supported by UC-MEXUS (University of California Institute for Mexico and the United States) CN-16-43, DGAPA-UNAM PAPIIT IN102716, and DGAPA-UNAM PASPA program.

%\bibliographystyle{siamplain}
%\bibliography{eliassi}

\end{document}